\documentclass[runningheads]{llncs}
\usepackage[lined,boxed,commentsnumbered,ruled,vlined,noend,linesnumbered,boxed]{algorithm2e}
\usepackage{float}
\usepackage{amssymb}
\usepackage{amsmath}
\usepackage{mathtools}
\usepackage{relsize}
\usepackage{todonotes}
\usepackage{hyperref}

\usepackage{framed}
\usepackage{url}
\usepackage{cite}
\usepackage{comment}
\usepackage{tcolorbox}
\tcbuselibrary{skins}
\usepackage{xcolor}
\colorlet{mix}{red!50!black}
\usepackage{hyperref}
\hypersetup{colorlinks={true},linkcolor={blue},citecolor=blue}
\newtheorem{rr}{Reducion Rule}
\usepackage{enumitem}
\usepackage{comment}
\usepackage{tikz}
\usepackage{epsfig}
\usepackage{wrapfig}
\usepackage{xspace}
\usetikzlibrary{shapes,backgrounds, patterns}
\usepackage{lipsum}
\usepackage[symbol]{footmisc}

\usepackage{comment}

\newcommand{\KFProb}[0]{{\sf KFVD}\xspace}

\title{An Improved  Exact Algorithm for Knot-Free Vertex Deletion}
\titlerunning{Knot-Free Vertex Deletion}

\author{Ajaykrishnan E S\inst{1}
\and Soumen Maity\inst{1}
\and Abhishek Sahu\inst{2}
\and  \\Saket Saurabh\inst{3,4} 
}

\authorrunning{E S.\,Ajaykrishnan et al.}

\institute{Indian Institute of Science Education and Research, Pune, India \and National Institute of Science Education and Research, An OCC of Homi Bhabha National Institute, Bhubaneswar, India \and The Institute of Mathematical Sciences, Chennai, India \and University of Bergen, Bergen, Norway\\
\email{\texttt{ajaykrishnan.es@students.iiserpune.ac.in}};
\email{\texttt{soumen@iiserpune.ac.in}};
\email{\texttt{abhisheksahu@niser.ac.in}};
\email{\texttt{saket@imsc.res.in}}\\
}

\colorlet{bscolor}{blue}

\DeclareUnicodeCharacter{B0}{\textless}

\begin{document}
\maketitle

\begin{abstract}

\newcommand{\bigoh}[0]{{\mathcal O}}

A knot $K$ in a directed graph $D$ is a strongly connected component of size at least two such that there is no  arc $(u,v)$ with $u \in V(K)$ and $v\notin V(K)$. Given a directed graph $D=(V,E)$, we study {\sc Knot-Free Vertex Deletion} (\KFProb), where the goal is to remove the minimum number of vertices such that the resulting graph contains no knots.  This problem naturally emerges from its application in deadlock resolution since knots are deadlocks in the OR-model of distributed computation. The fastest known exact algorithm in literature for {\KFProb}  runs in time $\mathcal{O}^\star(1.576^n)$. In this paper, we present an improved exact algorithm running in time $\mathcal{O}^\star(1.4549^n)$, where $n$ is the number of vertices in $D$. We also prove that the number of inclusion wise minimal knot-free vertex deletion sets is $\mathcal{O}^\star(1.4549^n)$ and construct a family of graphs with $\Omega(1.4422^n)$ minimal knot-free vertex deletion sets.\vspace*{2mm}
\end{abstract}

\section{Introduction}
The paper by Held and Karp \cite{10.1145/800029.808532} in the early sixties, sparked an interest in designing (fast) exact exponential algorithms. In the last couple of decades there has been immense progress in this field, resulting in non-trivial exact exponential algorithms \cite{10.1145/800029.808532, 10.1007/11602613_58, Razgon2007ComputingMD, lima2018and, 10.1007/978-3-030-64843-5_21, 10.1145/2897518.2897551, bjrklund:LIPIcs:2017:6948}. Alongside optimisation problems, it is also of interest to find exact exponential algorithms for enumeration problems. These are useful in answering natural questions that arise in Graph Theory, about the number of minimal (maximal) vertex subsets that satisfy a given property in a graph of order $n$. One can prove a trivial bound of $\mathcal{O}(2^{n}/\sqrt{n})$, but better bounds are known only for relatively few problems. The celebrated result of Moon, Moser \cite{moon_moser} proving an upper bound of $\mathcal{O^{\star}}(1.4422^n)$ for maximal independent sets, alongside that of Fomin et al. for Feedback Vertex Set\cite{10.1145/800029.808532} and Dominating Set\cite{10.1007/11602613_58} are some examples. We refer to the monograph by Fomin and Kratsch \cite{fomin_exact} for a detailed survey of the techniques and results in the field. 

A \emph{knot} in a directed graph $D = (V,E)$ is a strongly connected component of size atleast 2, without any out-edge. For a given input digraph $D$, the problem of finding the smallest vertex subset $S\subseteq V$ whose deletion makes $D$ Knot-Free is known as the {\sc Knot-Free Vertex Deletion} (KFVD) Problem. It finds applications in resolution of deadlocks in a classical distributed computation model referred to as the OR-model.

A distributed system consists of independent processors which are interconnected via a network which facilitates resource sharing. It is typically represented using a wait-for digraph $D = (V,E)$, where the vertex set $V$ represents processes and edge set $E$ represents wait-conditions. A deadlock occurs when a set of processes wait indefinitely for resources from each other to take actions such as sending a message or releasing a lock. Deadlocks \cite{10.1007/s10878-018-0279-5} are a common problem in multiprocessing systems, parallel computing and distributed systems.

The AND-model and OR-model are well investigated deadlock models in literature \cite{https://doi.org/10.1002/net.21537}. In the AND-model the process corresponding to a vertex $v$ in the wait-for graph can start only after the processes corresponding to its out-neighbours are completed, while in the OR-model $v$ can begin if atleast one of its out-neighbours has finished execution. Deadlocks in the AND-model correspond to cycles and hence deadlock prevention via preempting processes becomes equivalent to the {\sc Directed Feedback Vertex Set} problem in the wait-for graph, which has received considerable attention from various algorithmic viewpoints \cite{10.1145/1374376.1374404, Razgon2007ComputingMD, 10.1145/3446969, doi:10.1137/1.9781611976465.14}. Meanwhile, the deadlock prevention in OR-models which is equivalent to KFVD, has only been explored recently \cite{DBLP:conf/iwpec/BessyBCPS19, ramanujan2022exact}.

\subsection{Our Contributions}
In this paper, we build upon the work of Ramanujan et al. \cite{ramanujan2022exact} to obtain the following results. The algorithm is designed using the technique of {\em Branching} and its complexity is computed via {\em Measure \& Conquer}

\begin{theorem}
 There exists an algorithm for {\sc Knot-Free Vertex Deletion} running in $\mathcal{O}^*(1.4549^n)$. Furthermore, there is an infinite family of graphs for which the algorithm takes $\Omega(1.4422^n)$ time.
\end{theorem}

\begin{theorem}
 The number of inclusion-wise minimal knot-free vertex deletion sets is $\mathcal{O}^*(1.4549^n)$  and there exists an infinite family of graphs with $\Omega(1.4422^n)$ many such sets.
\end{theorem}

\noindent Our algorithm uses the simple observation that, {\em a directed graph $D$ is knot-free if and only if every vertex in $D$ has a path to some sink}. Furthermore, for any given directed graph $D$ in order to find the optimal solution $S$ to the $KFVD$ problem, it suffices to find the set of sinks $Z$ in $D-S$ \cite{DBLP:conf/iwpec/BessyBCPS19}, since $Z = N^+(Z)$. Given an instance $I = (D,V_1,V_2)$, we pick a vertex $v\in V_1$ and branch on the possibility that $v$ is a sink vertex in some optimal solution. The measure of the instance is $\phi(I) = \lvert V_1\rvert + \frac{\lvert V_2 \rvert}{4}$. The elements of $V_2$ are non-sink vertices and whenever we conclude that a vertex in $V_1$ will not be a sink in the final solution, we shift it to $V_2$, giving us a potential drop. This is how we capture the drop in measure in the branch where $v$ is not a sink. In the other branch where $v$ is a sink, we prove that $N^+(v)$ must belong to the solution set. We further prove that the set of vertices that can reach $v$ after $N^+(v)$ is deleted, cannot contain solution vertices and can hence be removed. Finally, we also show that if $v$ can reach $u$ after $N^+(u)$ is deleted then $u$ cannot be a sink vertex in any optimal solution corresponding to this branch. The above observations allow us to delete vertices or move them from $V_1$ to $V_2$ which gives a potential drop. We do a "potential sensitive" branching, followed by an extensive case analysis in order to obtain the final algorithm. 

\subsection{Previous Work}

Consider the {\sc Directed Feedback Vertex Set} problem, where one has to find the smallest size vertex subset whose deletion makes the input digraph acyclic. Observe that acyclic graphs are also knot-free since they do not have strongly  connected components of size greater than $1$. Hence one way of solving KFVD is to use the solution obtained via an algorithm for DFVS, since the problem has received considerable attention in the past\cite{10.1145/1374376.1374404, 10.1145/3446969, Razgon2007ComputingMD, doi:10.1137/1.9781611976465.14}. But the size of the solution for KFVD could be considerably smaller than that of DFVS and hence it is of merit to focus on the problem directly. 

In 2019, Carnerio, Protti and Souza \cite{10.1007/s10878-018-0279-5} first investigated deadlock resolution in various models via arc or vertex deletion. It was shown that corresponding to the OR-model, the arc deletion problem is polynomial time solvable, while the vertex deletion (KFVD) remains NP-complete even for graphs of maximum degree $4$. They further showed that KFVD is solvable in $\mathcal{O}(m\sqrt{n})$ time in graphs of maximum degree $3$. The first and only non-trivial exact exponential algorithm for KFVD was presented by Ramanujan et al. \cite{ramanujan2022exact} and runs in time $\mathcal{O}^{\star}(1.576^n)$. From the aspect of Parameterised Complexity, Bessy et al. \cite{DBLP:conf/iwpec/BessyBCPS19} showed that KFVD parameterised by solution size is $W[1]$-hard but FPT parameterised by clique-width. The fine grained complexity of $KFVD$ with respect to treewidth ($2^{\mathcal{O}(tw)}$) given ETH and size of largest strongly connected component ($2^{\phi}n^{\mathcal{O}(1)}$) given SETH is also know via \cite{DBLP:conf/iwpec/BessyBCPS19, knot_finegrained}.\\

\noindent {\em Organization of the paper:} In section 2, we state and prove some theorems which are used throughout the paper. By section 3, we present the algorithm and prove its correctness in section 4. We analyse the time complexity in section 5 and provide a lower bound for the run time of the algorithm in section 6. Finally, section 7 is dedicated to proving the upper and a complimentary lower bound on the number of minimal KFVD sets in a graph of size $n$.
\section{Preliminaries  and Auxiliary Results}
In this section we state some commonly used definitions, notations and useful auxiliary results. We also formalize the potential function which is integral to our algorithm and finally state and prove a few reduction rules which are used throughout the paper.\\

\noindent 
{\bf Notation:}  
For a set $S \subseteq V(D)$, $G[S]$ denotes the subgraph of $D$ induced on $S$ and $G[D-S]$ denotes the subgraph induced on $V(D)\setminus S$. A {\em path} $P_{(u,w)}$ from $u$ to $w$ of length $\ell$ is a sequence of distinct vertices $v_1,v_2,\ldots,v_{\ell}$ such that $(v_i,v_{i+1})$ is an arc, for each $i,\ i\in [\ell-1]$ and $v_1=u$, $v_\ell=w$. We define the \emph{in-reachability set} of a vertex $v$ denoted by $R^-(v)$, as the set of vertices that can reach $v$ via some directed path in $D - N^+(v)$. Notice that $v\in R^-(v)$. 
We define $R(v)=N^+(v)\cup R^-(v)$. For graph-theoretic terms and definitions not stated explicitly here, we refer to \cite{diestel-book}. 

\subsection{Auxiliary Results}
In this subsection, we first state some of the known reduction rules and some new ones that we use in our branching algorithm. 

\begin{proposition}{\rm\cite{DBLP:conf/iwpec/BessyBCPS19}\label{prop:1}}
    A digraph $D$ is knot-free if and only if for every vertex $v$ of $D$, $v$ has a path to a sink.
\end{proposition}

\begin{corollary}{\rm\cite{DBLP:conf/iwpec/BessyBCPS19}\label{cor1}}
    For any minimal solution $S\subseteq V(D)$ with the set of sink vertices $Z$ in $D-S$, we have $N^+(Z) = S$. 
\end{corollary}
  
Proposition \ref{prop:1} and Corollary \ref{cor1} imply that given a digraph $D$, the problem of finding a set of sink vertices $Z$ such that every vertex in $V(D)-N^+(Z)$ has a directed path to a vertex in $Z$ and  $|N^+(Z)|$ is minimum; this is equivalent to the \textsc{Knot-Free Vertex Deletion} ({\KFProb}) problem. Therefore, our algorithm aims to find the set of sink vertices $Z$ corresponding to an optimal solution, while minimizing $|N^+(Z)|$ instead of directly finding the deletion set. \\

\noindent
{\bf Strategy of our Algorithm.}  The algorithm expands on the ideas used in \cite{ramanujan2022exact} where the algorithm branches on the possibility that a vertex $v\in V(D)$ is either a sink or a non-sink vertex in some optimal solution. In the branch where we conclude $v$ to be a non-sink vertex there are two possibilities, $v$ is either in the deletion set or not. To track this additional information that $v$ is non-sink, we use a potential function $\phi$ for $V(D)$ defined as follows.
  
\begin{definition}[Potential function]
    Given a digraph $D=(V,E)$, we define a \emph{potential function} on $V(D)$, $\phi:V(D)\rightarrow \{0.25,1\}$ such that $\phi(v)=1$, if $v$ is a potential vertex to become a sink in an optimal solution and $\phi(v)=0.25$, if $v$ is a non-sink vertex. For any subset $V'\subseteq V(D)$, $\phi(V')=\sum_{x\in V'} \phi(x)$. We call a vertex $v$ as an \emph{undecided} vertex if $\phi(v)=1$ and a \emph{semi-decided} vertex if $\phi(v)=0.25$. 
\end{definition}

\begin{definition}[Feasible solution]
    A set $S\subseteq V(D)$ is called a feasible solution for $(D,\phi)$ if $D-S$ is knot-free and for any sink vertex $s$ in $D-S$, $\phi(s)=1$. ${\KFProb}(D,\phi)$ denotes the size of an optimal solution for $(D,\phi)$. 
\end{definition}

\noindent To solve the \textsc{Knot-Free Vertex Deletion} problem on a digraph $D$, we initialize the potential values of all vertices to $1$. As soon as we decide a vertex to be a non-sink vertex, we drop its potential by $0.75$. Any vertex whose potential is $0.25$ cannot become a sink in the final knot-free graph corresponding to an optimal feasible vertex deletion set of $(D,\phi)$.

\noindent Now we are ready to define the \emph{out-reachability set} of $v$ which is the set of \emph{undecided} vertices which are reachable from $v$ even after deleting their out-neighbours, denoted as $R^+(v)$. Note that, $u \in R^+(v)$ if and only if $v\in R^-(u)$. We shall use $S_{min}$ to denote a minimal solution for the given instance $(D,\phi)$ and $Z_{min}$ for the set of sinks in $D-S_{min}$. Similarly, we use $S_{opt}$ and $ Z_{opt}$ to refer to an optimal solution and its sink set respectively. We further define an update function to make our algorithm description concise.\\
\vspace{-20pt}
\begin{algorithm}[ht!]
	\SetAlgoLined
	\SetKwData{I}{I}\SetKwData{size}{size}\SetKwData{clrr}{cLRR}\SetKwData{Stop}{Stop}\SetKwData{sizec}{size(cLRR)}
	\small 
	\vspace{4pt}
    \KwIn{A directed graph $D$, a potential function $\phi$, a sink vertex $s$ (optional) and a set  of non-sink vertices $NS$ (optional)}
    \KwOut{ An updated digraph $D'$ and an updated potential function $\phi'$}\vspace*{2mm}
    \hrule
    \label{helper1}\vspace*{2mm}
    \caption{$update(D,\phi,s,NS)$}
    $D' = D$, $\phi' = \phi$\\
    \If {\text{input} s \text{is provided}} {
        $D'= D' - R(s)$\\
        \For{v $\in R^+(s)$ }{
            $\phi'(v) = 0.25$}}
    \If {input NS is provided} { 
        \For{v $\in$ NS}{
            $\phi'(v) = 0.25$}}
    \Return{$D', \phi'$}\vspace*{2mm}
\end{algorithm}\\[-37pt]

\begin{rr}{\rm\cite{ramanujan2022exact}}\label{rr1}
    If all the vertices in $D$ are semi-decided and $D$ has no source or sink vertices, then $V(D)$ is contained inside any feasible solution for $(D,\phi)$. 
\end{rr}

\begin{rr}{\rm\cite{ramanujan2022exact}}\label{rr2} 
    Let $v\in D$ be such that $N^-(v)=\emptyset$, $S$ is an optimal solution for $D$ iff $S$ is an optimal solution for $D' = D-v$.
\end{rr}

\begin{rr}{\rm\cite{ramanujan2022exact}}\label{rr3} 
    Let $v\in D$ be such that $N^+(v)=\emptyset$, $S$ is an optimal solution for $D$ iff $S$ is an optimal solution for $D' = D-R(v)$.
\end{rr}
 
\begin{proposition}\label{claim:1} 
    If  $x \in Z_{min}$, then $N^+(x) \subseteq S_{min}$, $S_{min}\cap R^-(x)=\emptyset$ and $Z_{min} \cap R^+(x) = \emptyset$.
\end{proposition}
\begin{proof}
    By the definition of a sink vertex, if $x \in Z_{min}$ then $N^+(x)$ is in $S_{min}$. Let $Y= S_{min}\cap R^-(x)$. We claim $S'=S_{min}\setminus Y$ is also a solution which will contradict the fact that $S_{min}$ is a minimal solution. Suppose $S'$ is not a solution, then there exists a vertex $v$ in $D-S'$, that does not reach a sink in $D-S'$ but $v$ reaches some sink $s$ in $D-S_{min}$. Since every vertex in $Y$ reaches sink $x$ in $D-S'$, $v\notin Y$. If $s$ is not a sink vertex in $D-S'$, then some vertex $y\in Y$ is an out-neighbor of $s$. But then $v$ can reach the sink $x$ in $D-S'$ via $y$. Hence, $S'$ is a solution of size strictly smaller than $S_{min}$, which is a contradiction. Since $S_{min}\cap R^-(x)=\emptyset$, we also have $Z_{min}\cap R^-(x)=\emptyset$. Now, let $s\in R^+(x)$. If $s\in Z_{min}$ then by definition of $R^+(x)$, $x\in Z_{min}\cap R^-(s)$ which is a contradiction. \qed
\end{proof}

\begin{proposition}\label{prop:2}
    If  $x \in Z_{min}$, then $S_{min}=N^+(x)\cup S'$, where $S'$ is a minimal feasible solution for $(D',\phi') = update(D,\phi,x,-)$. Also, if $S_{min}'$ is a minimal solution for $(D',\phi')$ then $S=N^+(x)\cup S_{min}'$ is a solution for $(D,\phi)$.
\end{proposition}
\begin{proof}
    First, assume that $S'_{min}$ is a minimal solution for $(D',\phi')$. We claim that $S=S'_{min}\cup N^+(x)$ is a solution for $(D,\phi)$. Suppose not, then there exists a vertex $v$ that does not reach a sink in $D - S$. Note that $v\notin R^-(x)$ as all vertices in $R^-(x)$ can reach sink $x$ in $D - S$. Then $v\notin R(x)\cup S$ and hence it is also in $D'-S_{min}$, where it reaches a sink $s$. Since this path is disjoint from $S'_{min}\cup R^{-}(x)$, $v$ still can reach $s$ via the same path in $D - S$. Note that $s$ is not a sink in $D - S$ only if it has an out-neighbor in $R^-(x)\setminus N^+(x)$. But then $s$ and $v$ are in $R^-(x)$ which is a contradiction. Hence, $v$ can still reach the same sink $s$ and $S=S'_{min}\cup N^+(x)$ is a solution for $D$. 
    
    Now, we prove that $S'=S_{min}\setminus  N^+(x)$ is a solution for $(D',\phi')$. 
    Since, $S_{min}$ is optimal, $S'\cap R^+(x)=\emptyset$ via Claim \ref{claim:1} and thus, $S'$ is feasible. If $S'$ is not a solution, then there has to be a vertex $v$ in $D'-S'$ that does not reach any sink. Since $S_{min}\subseteq R(x)\cup S'$, $v\notin S_{min}$. But $S_{min}$ is a solution for $(D,\phi)$ and $v$ can reach some sink $s \in Z_{min}$ in $D-S_{min}$. 
    Note that $s\neq x$, since $v$ is not in $R^-(x)$ and cannot reach $x$ in $D-S_{opt}$. Then $v$ has a path to $s$ which is disjoint from $S_{opt}\supseteq N^+(x)$. Moreover this path is also disjoint from the set $R(x)\setminus N^+(x)$, since $v\notin R(x)$. 
    Hence this path is disjoint from $R(x)$ and $S_{opt}$. Therefore, $v$ can reach $s$ via the same path in $D'-S'$. If $s$ is a sink in $D-S_{opt}$, then $s$ is also a sink in $D'-S'$. Hence, $v$ has a path to a sink in $D'-S'$, which is a contradiction. It implies that  $S'=S_{min}\setminus  N^+(x)$ is a solution for $(D',\phi')$.
    
    Finally, assume $S'=S_{min}\setminus  N^+(x)$ is not a minimal solution for $(D',\phi')$. Then there exists $S''\subset S'$ which is also a solution for $(D',\phi')$, but then $S''\cup N^+(x)$ is also a solution for $(D,\phi)$ contradicting the minimality of $S_{min}$.\qed
\end{proof}

\begin{corollary}\label{cor2}
    If  $x \in Z_{opt}$, then $|S_{opt}|=  |N^+(x)|+{\KFProb}(D',\phi')$, where $(D',\phi') = update(D,\phi,x,-)$.
\end{corollary}
\begin{proof}
    Let us assume $S_{opt}'$ to be an optimal solution to $(D',\phi')$. By Proposition \ref{prop:2}, $S = N^+(x)\cup S_{opt}'$ is a solution to $(D,\phi)$. Now, if $(D,\phi)$ has a minimal solution $\Tilde{S}$ smaller than $S$ then via Proposition \ref{prop:2}, $\Tilde{S}\setminus N^+(x)$ is a solution to $(D',\phi')$ that is smaller than $S_{opt}'$. This is a contradiction and hence the claim is true.\qed
\end{proof}

Now, we define a \emph{drop function} for every undecided vertex, denoted by $\psi(x)$ which takes the value $\psi (x) = \phi(R(x)) + 0.75*|R^+(x)\setminus R(x)|$. This function keeps track of the drop in potential when $x$ becomes a sink in some branch. Further, we define the \emph{surviving set} of vertex $x$, denoted by $\mathcal{S}_x = \{u\in N^+(x) | N^-(u)\cap \mathcal{U} = \{x\}\}$. Observe that for a given instance $(D,\phi)$, if $x\in \mathcal{U}$ has a nonempty survivor set, then the vertices of $\mathcal{S}_x$ is in a minimal solution $S_{min}$ if and only if $x$ is in $Z_{min}$ since Corollary \ref{cor1} requires $S_{min} = N^+(Z_{min})$ and $x$ is the only in-neighbour of elements of $\mathcal{S}_x$ which can be in $Z_{min}$. Finally, we have \emph{candidate sink set} of vertex $x$, denoted by $\mathcal{C}_x = \{u\in R^+(y)\ |\ y\in N^+(x)\}$. Observe that for any given instance $(D,\phi)$ if $x\notin Z_{min}$ then some out-neighbour $y$ of $x$ must satisfy $y\notin S_{min}$. This $y$ must reach some sink $s$ in the final solution, and by definition, $s$ must belong to $R^+(y)$. Thus $\mathcal{C}_x$ is the set of sinks the out neighbours of $x$ can possibly reach if $x$ is not a sink in the final solution. Also, note $\mathcal{C}_x \subseteq N^-(x)\cup R^+(x)$, since given $y\in N^+(x)$ and $s\in R^+(y)$, either $x\in N^+(s)$ or $x$ has a path to $s$ via $y$ in $D-N^+(s)$.

\section{An algorithm to compute minimum knot-free vertex deletion set}

In this section, we provide an exact exponential algorithm (Algorithm \ref{alg_1}) to compute the minimum knot-free vertex deletion set. We begin by initialising the potential of all vertices to $1$ and as soon as we decide a vertex to be non-sink we reduce its potential to $0.25$. In our algorithm, whenever we encounter a sink or a source, we remove them using Reduction Rules \ref{rr2} and \ref{rr3}. If all the vertices have potential $0.25$, then we apply Reduction Rule \ref{rr1} to solve the instance in polynomial time. 

At any point, if there exists an undecided vertex $x$ of potential $\psi(x)\geq 3.75$ then we branch on the possibility of it being a sink or non-sink in the optimal solution. Here we benefit from the high potential drop in the branch where $x$ becomes a sink which gives us a branching factor of $(3.75,0.75)$. Once such vertices are exhausted, we choose an undecided vertex $x$ with the \emph{maximum} number of undecided neighbours to branch on. We show that if $x$ is not a sink in the optimal solution then some other vertex $s$ from $\mathcal{C}_{x}$ has to be. Further, the bound on $\psi(x)$ helps us limit the cardinality of $\mathcal{C}_{x}$ and consequently the number of branches. In this case, we get a set of vertices from which atleast one has to be a sink in the final solution. Since, each branch has some vertex becoming a sink, the potential drop will be high enough to give a \emph{good} running time for our algorithm.
\begin{algorithm}[ht!]	\label{alg_1}
    \SetAlgoLined
    \SetKwData{I}{I}
    \SetKwData{size}{size}
    \SetKwData{clrr}{cLRR}
    \SetKwData{Stop}{Stop}
    \SetKwData{sizec}{size(cLRR)}
    \small 
    \vspace{6pt}
    \KwIn{A directed graph $D$ and a potential function $\phi$ }
    \KwOut{ The size of a minimum knot-free vertex deletion set}\vspace*{3mm}
    \hrule
    \label{algo1}\vspace*{3mm}
    \caption{\KFProb($D,\phi$)}
    
    \If {$\exists\ x$ such that $N^-(x)=\emptyset$} {
        \Return{$\KFProb(D-\{x\},\phi)$;}}\vspace*{2mm}
    \If {$\exists\ x$ such that $N^+(x)=\emptyset$} { 
        \Return{ $\KFProb(D-R(x),\phi)$;}}\vspace*{2mm}
    \If {$V(D) \subseteq\ \Bar{\mathcal{U}}$} {
        \Return{ $|V(D)|$;}}\vspace*{2mm}
    \If {$\exists\ x\in \mathcal{U}$ such that $\psi(x)\geq 3.75$}{
        $D_1, \phi_1 = update(D,\phi,x,-)$;\\
        $D_2, \phi_2 = update(D,\phi,-,x)$;\\
        \Return{ $\min\{\KFProb(D_1,\phi_1)+|N^+(x)|,\ \KFProb(D_2,\phi_2)\}$;}}\vspace*{2mm}

    pick $x\in \mathcal{U}$ maximizing $\lvert N(x)\cap\mathcal{U}\rvert$\vspace*{2mm}
            
    \If {$\mathcal{S}_x = \emptyset$}{  
        \For{$s_i \in\mathcal{C}_x$}{
            $D_i,\phi_i = update(D,\phi,s_i,-)$}
        \Return{ $\min\{\KFProb(D_x,\phi_x)+|N^+(x)|,\ \min_{s_i}\{\KFProb(D_i,\phi_i)+|N^+(s_i)|\}\}$;}}
        
    \Else {
        \If {$\lvert N(x)\cap\mathcal{U}\rvert \geq 1$} {
            $D_x,\phi_x = update(D,\phi,x,-)$;\\
            pick $y \in\mathcal{S}_x$\\
            set $B_x = R^+(y)$\\
            \For{$s_i \in  B_x$}{
                \If{$N^+(s_i)\subseteq N^+(s_j)$ for some $s_j\in B_x$}{
                    $B_x = B_x-s_j$}}
            \For{$s_i \in  B_x$}{
                $NS_i = \{s_j : j<i\}$\\
                $D_i,\phi_i = update(D,\phi,s_i,NS_i)$}  
            \Return{ $\min\{\KFProb(D_x,\phi_x)+|N^+(x)|,\ \min_{s_i}\{\KFProb(D_i,\phi_i)+|N^+(s_i)|\}\}$;}}
        \Else {
            $D_x,\phi_x = update(D,\phi,x,-)$;\\
            pick $y \in\mathcal{S}_x$, $s \in R^+(y)$\\
            $D_s,\phi_s = update(D,\phi,s,-)$\\
            \Return{ $\min\{\KFProb(D_x,\phi_x)+|N^+(x)|,\ \KFProb(D_s,\phi_s)+|N^+(s)|\}\}$;}}}\vspace*{6pt}
\end{algorithm}

\section{Correctness of the Algorithm}
In this section we prove that  Algorithm KFVD returns an optimal knot-free vertex deletion set for any input instance. Let \emph{Subroutine $0$} represent the lines $1-6$ of the algorithm. The correctness of lines $5-6$, $1-2$ and $3-4$ follow from Reduction Rules \ref{rr1}, \ref{rr2} and \ref{rr3} respectively. We use the terms \emph{Subroutine $1$}, \emph{Subroutine $2$}, \emph{Subroutine $3$} and \emph{Subroutine $4$} to refer to lines $7$-$10$, $12$-$15$, $17$-$27$ and $28$-$32$ of Algorithm \KFProb~respectively. The correctness of the subroutines are verified in the following sections.
    \subsection {Correctness of Subroutine 1} \label{sub1}
    In Subroutine $1$ we branch on an undecided vertex $x\in V(D)$, which has $\psi(x) \geq 3.75$. The following lemma proves the correctness of the subroutine. 

\begin{lemma}\label{lem:subroutine1}
    If $\ \exists\ x\in V(D)$ such that $x\in\mathcal{U}$ and $\psi(x)\geq 3.75$, then $\KFProb(D,\phi)=\min\{\KFProb(D_1,\phi_1)+|N^+(x)|,\text{ }\KFProb(D_2,\phi_2)\}$ where, $(D_1,\phi_1) = update(D,\phi,x,-)$ and $(D_2,\phi_2) = update(D,\phi,-,x)$.
\end{lemma}
\begin{proof}
    We prove the lemma using an inductive argument on the potential of the instance $(\phi(D))$. Observe that if we consider the base case $\phi(D) = 1$, then there is only one undecided vertex in the input instance and the recurrence holds true. Now given an instance $D$, assume that the algorithm computes the correct solution for all smaller instances. Let the solutions for $\KFProb(D_1,\phi_1)$ and $\KFProb(D_2,\phi_2)$ be $S_1$ and $S_2$, respectively. Assuming $S_{opt}$ is an optimal solution for $\KFProb(D,\phi)$, we evaluate the two possibilities: 
    \begin{itemize}
        \item \textbf{Case 1:} $\boldsymbol{x\in Z_{opt}}$.
           We show that $S_1\cup N^+(x)$ is an optimal solution for  ${\KFProb(D,\phi)}$ if and only if  $S_1$ is an optimal solution for ${\KFProb(D_1,\phi_1)}$. The arguments are exactly the same as that in Corollary \ref{cor2}. We also claim that $\KFProb(D_2,\phi_2) \geq \KFProb(D_1,\phi_1)+|N^+(x)| $. For contradiction, suppose that is not the case, then any optimal solution $S_2$ for ($D_2,\phi_2$) is also a feasible solution for $\KFProb(D,\phi)$. But $S_2$ has size strictly smaller than $S_{opt}$, which contradicts our assumption that it is optimal. Hence, $\KFProb(D,\phi)=$ $\min\{\KFProb$ $(D_1,\phi)+|N^+(x)|,\text{ }\KFProb(D,\phi_2)\}$.\vspace*{2mm}
        \item \textbf{Case 2:} $\boldsymbol{x\notin Z_{opt}}$.
            In this case, $\KFProb (D_2,\phi_2)=\KFProb(D,\phi)$, by definition. Also $\KFProb(D_2,\phi_2) \leq \KFProb(D_1,\phi_1)+|N^+(x)| $, otherwise we have $S'=S_1\cup N^+(x)$ as a solution with size strictly smaller than $S_{opt}$ which contradicts our assumption that it is optimal. Hence, $\KFProb(D,\phi)=$ $\min\{\KFProb(D_1,\phi)+|N^+(x)|,\text{ } \KFProb(D,\phi_2)\}$.
    \end{itemize}\qed
\end{proof}

In Subroutine $1$, we get a branching vector $(3.75,0.75)$. After exhaustively running Subroutine $1$, every remaining undecided vertex satisfies $\psi(x)\leq 3.75$. An important implication of this bound is the restricted number of undecided vertices in $N(x)$ as well as $\mathcal{C}_x$, which we prove in the 
 following lemma.

\begin{lemma}\label{lem:neighbourhood}
    If $(D,\phi)$ is an instance on which Subroutine 1 is no longer applicable, then every $x\in\mathcal{U}$ satisfies, $|\mathcal{C}_x|\leq 2$.
\end{lemma}
\begin{proof}
    If Subroutine 1 is no longer applicable to $(D,\phi)$, then every undecided vertex has at most 2 undecided neighbours, since $N(x)\subseteq R(x)$ and $\phi(R(x))$ is counted towards $\psi(x)$.\\
    To begin with, consider the possibility that, $|R^+(x)\setminus R(x)| \geq 3$. Observe that while computing $\psi(x)$, vertices of $|R^+(x)\setminus R(x)|$ contribute a total of $2.25$, the potential of $x$ contributes $1$ and the in-neighbour and out-neighbour of $x$ has to contribute atleast $0.5$. This adds up to a total of $3.75$ which is not allowed since Subroutine 1 is no longer applicable. Hence $|R^+(x)\setminus R(x)| \leq 2$. Also, recall that $\mathcal{C}_x\subseteq N^-(x)\cup R^+(x)\subseteq R(x)\cup R^+(x)$. Now we can consider the following possibilities.
    \begin{itemize}
        \item \textbf{Case 1:} $\boldsymbol{|R^+(x)\setminus R(x)| = 0}$. Here $\mathcal{C}_x\subseteq R(x)$. Also $\phi(R(x))$ is less than $3.75$, out of which $x$ contributes 1. Hence $R(x)$ can have at most 2 more undecided vertices and consequently, $|\mathcal{C}_x| \leq 2$.\vspace{-5pt}\\
        \item \textbf{Case 2:} $\boldsymbol{|R^+(x)\setminus R(x)| = 1}$. In this case, since $R^+(x)\setminus R(x)$ contributes $0.75$, the contribution of $R(x)$ to  $\psi(x)$ has to be less than $3$. Since $x$ itself contributes 1, we can have at most one other undecided vertex in $R(x)$. Hence $|\mathcal{C}_x| \leq 2$.\vspace{-5pt}\\
        \item \textbf{Case 3:} $\boldsymbol{|R^+(x)\setminus R(x)| = 2}$. Here, vertices in $R^+(x)\setminus R(x)$ contribute a total of $1.5$ to $\psi(x)$ which gives $\phi(R(x)) < 2.25$. Hence if $|\mathcal{U}\cap R(x)|\geq 2$ then we get $|R(x)| = 2$ which is not possible since $d^+(x), d^-(x)\geq 1$. Hence $\mathcal{U}\cap R(x)=\{x\}$ and $|\mathcal{C}_x| \leq 2$.
    \end{itemize}\qed
\end{proof}
    \subsection{Correctness of Subroutine 2} \label{sec4:subsec2}
    In Subroutine 2, we deal with undecided vertices which has an empty surviving set corresponding to them. The fact that every out-neighbour of such a vertex has some other undecided in-neighbour helps us establish a high potential for elements of $\mathcal{C}_x$. 

\begin{lemma}\label{lem:subroutine2}
    Let $(D,\phi)$ be an instance such that $\psi(v) < 3.75$, for every vertex in $\mathcal{U}$. let $x$ be a vertex maximizing $\mathcal{U}\cap N(x)$ and $\mathcal{S}_x = \emptyset$. We claim, $$\KFProb(D,\phi) = min\{\KFProb(D_x,\phi_x)+|N^+(x)|,\ min_{s_i}\{\KFProb(D_i,\phi_i)+|N^+(s_i)|\}\}$$ where $(D_x,\phi_x) = update(D,\phi,x,-)$ and $(D_i,\phi_i) = update(D,\phi,s_i,-)$ for every $s_i\in\mathcal{C}_x$.
\end{lemma}
\begin{proof}
    In any given minimal solution, either $x$ is a sink or at least one of its out-neighbours must survive. The out-neighbour which survives, say $y$, must be able to reach a sink in the final solution. This sink by definition, belongs to $R^+(y)$. Hence, if $x$ is not a sink in the final solution, then at least one vertex of $\mathcal{C}_x$ has to be. Thus, every possible minimal solution contains at least one element from the set $\{x\}\cup\mathcal{C}_x$ in its sink set. By induction, assuming that $\KFProb(D',\phi')$ returns the optimal solution for every instance smaller than $(D,\phi)$ along with Corollary \ref{cor2}, proves the lemma.\qed
\end{proof}

\noindent Observe that, since $|\mathcal{C}_x|\leq 2$, we have at most 3 branches when running Subroutine 2. Further we claim that whenever we branch on the case where $s_i\in\mathcal{C}_x$ becomes a sink, the potential drop is at least $3$. We have, either $x\in N^+(s_i)$ or $x\in R^-(s_i)$. Further, $s_i\in R^+(y)$ for some $y\in N^+(x)$ and since $\mathcal{S}_x = \emptyset$, $y$ has some undecided in-neighbour $z$ different from $x$. Note that, $z\neq s_i$ since $s_i\in R^+(y)$. Similar to $x$, $z\in N^+(s_i)$ or $z\in R^-(s_i)$. Hence $\psi(s_i)\geq\phi(\{s_i,x,z\})=3$. Now let us look at the possible branches which could arise and their corresponding worst case branching factor. Note that if $|\mathcal{C}_x| = 0$, no branching is involved and the subroutine is executed as a reduction rule. 
\begin{itemize}
    \item \textbf{Case 1:} $\boldsymbol{\mathcal{C}_x = \{s\}}$. Here, we have the following possibilities:\\
    \begin{itemize}
        \item $|N(x)\cap\mathcal{C}_x| = 0$, in which case, $N[x]$ has potential at least $1.5$, since $x$ has atleast two neighbours
        and $\phi(x) = 1$. Further, $s\in\mathcal{C}_x$ contributes $0.75$, giving $\psi(x)\geq 2.25$.
        \begin{figure}
            \centering
            \includegraphics[height=1.25in,width=2.25in]{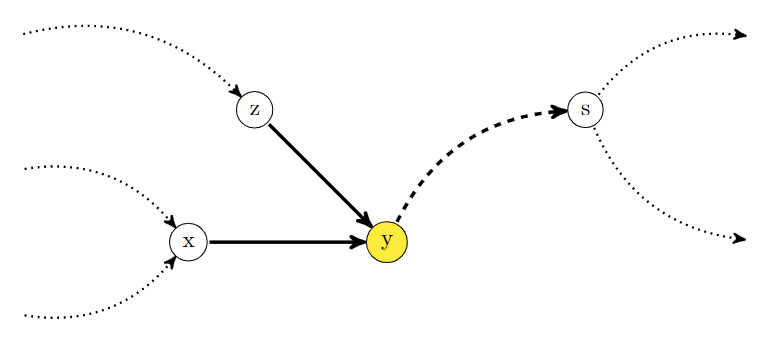}
            \caption{Case 1.1: $N[x]\cap\mathcal{C}_{x}=\emptyset$}
            \label{S2_C1_1}
        \end{figure}\footnote{White vertices are undecided, yellow ones are semi-decided. Dotted lines indicate possibility of in and out-neighbours, dashed lines denote directed paths and thick lines denote edges.}
        \item $|N(x)\cap\mathcal{C}_x| = 1$, in which case, $N[x]$ has potential at least $2.25$, since $\phi(x)=\phi(s)=1$, and $x$ has at least one more neighbour, giving $\psi(x)\geq 2.25$.\\
        \begin{figure}
            \centering
            \includegraphics[height=1.3in,width=2.25in]{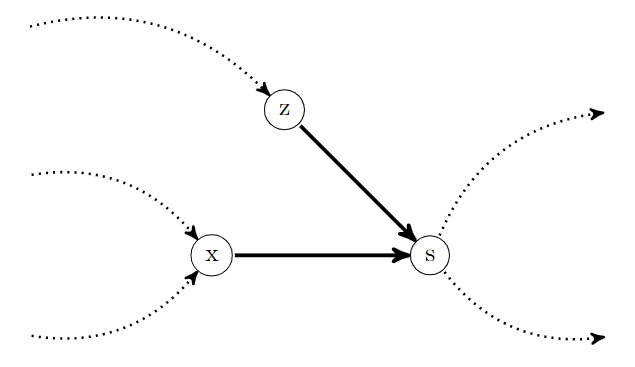}
            \caption{Case 1.2: $s\in N[x]\cap\mathcal{C}_{x}$}
            \label{S2_C1_2}
        \end{figure}
    \end{itemize}
    \newpage
    \item \textbf{Case 2:} $\boldsymbol{\mathcal{C}_x = \{s_1, s_2\}}$. Here the possibilities are as follows.\\
    \begin{itemize}
        \item $|N(x)\cap\mathcal{C}_x| = 0$, in which case, $N[x]$ has potential at least $1.5$, since $x$ has at least two neighbours and $\phi(x) = 1$. Further $s_1, s_2\in\mathcal{C}_x$ contribute $1.5$, giving $\psi(x)\geq 3$.
        \begin{figure}
            \centering 
            \includegraphics[height=1.20in,width=1.85in]{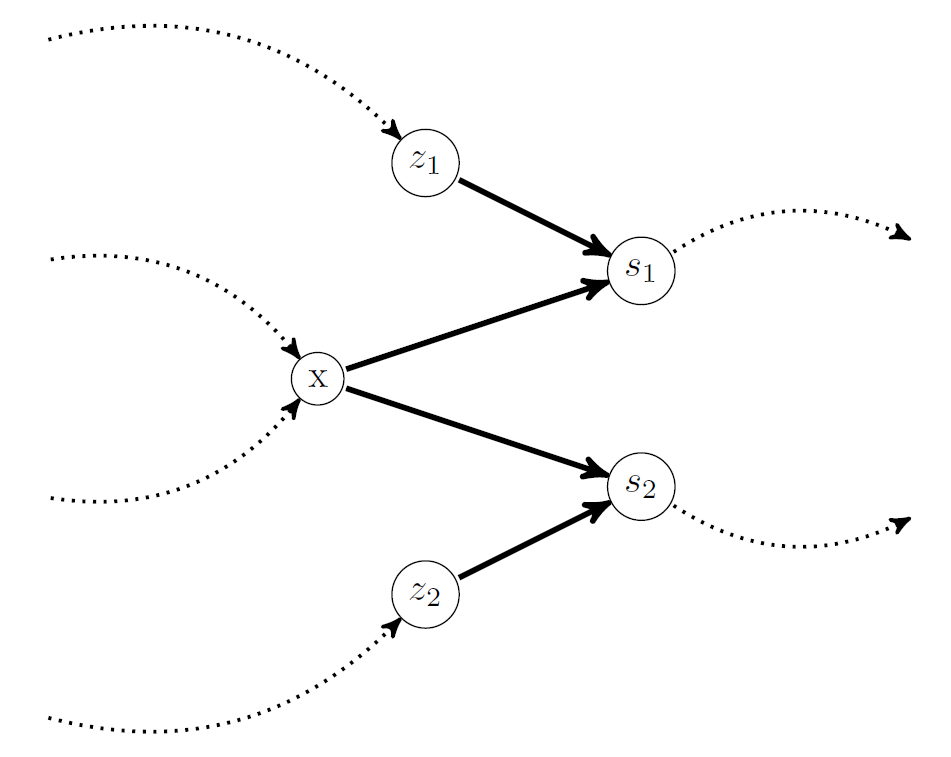}
            \caption{Case 2.1: $s_1,s_2\in N[x]\cap\mathcal{C}_{x}$}
            \label{S2_C2_3}
        \end{figure}
        \item $N(x)\cap\mathcal{C}_x = \{s_1\}$, in which case, $N[x]$ has potential at least $2.25$, since $x$ has at least two neighbours, including $s_1$. Further $s_2\in\mathcal{C}_x\setminus N(x)$ contributes $0.75$, giving $\psi(x)\geq 3$.
        \begin{figure}
            \centering
            \includegraphics[height=1.20in,width=1.85in]{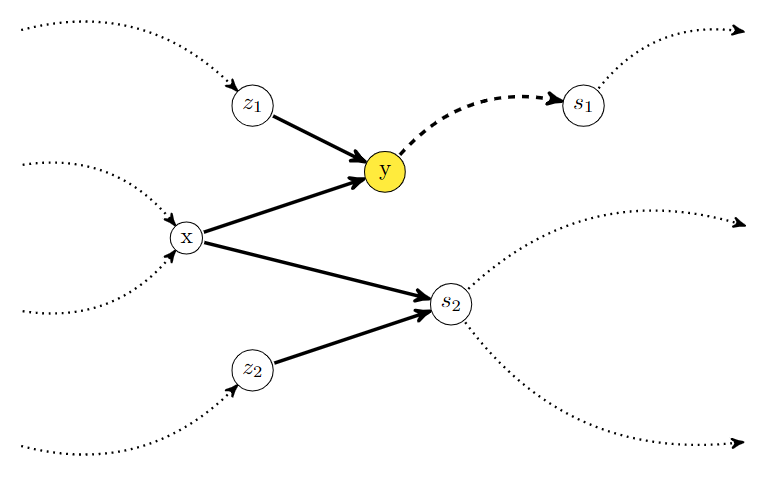}
            \caption{Case 2.2: $s_1\in N[x]\cap\mathcal{C}_{x}$}
            \label{S2_C2_4}
        \end{figure}
        \item $|N(x)\cap\mathcal{C}_x| = 2$, in which case, $N[x]$ has potential at least $3$, due to $\phi(x)=\phi(s_1)=\phi(s_2)=1$ and hence $\psi(x)\geq 3$.
        \begin{figure}
            \centering
            \includegraphics[height=1.20in,width=2.25in]{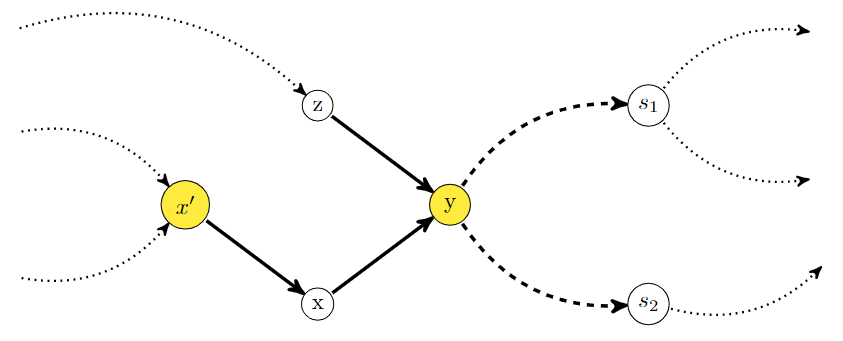}
            \caption{Case 2.3: $N[x]\cap\mathcal{C}_{x}=\emptyset$}
            \label{S2_C2_5}\vspace{-20pt}
        \end{figure}
    \end{itemize}
\end{itemize}
\noindent Hence the worst case branching vectors from Subroutine 2 are $(3,2.25)$ and $(3,3,3)$. 
    \subsection{Correctness of Subroutine 3} \label{sec4:subsec3}
    In Subroutine $3$ we branch on undecided vertices with a non-empty surviving set and at least one undecided neighbour. Let us define $D_x = \{u\ |\ \exists\ v\in R^+(y),\ v\neq u,\ N^+(v)\subseteq N^+(u)\}$ and $B_x$ to be $R^+(y)\setminus D_x$. We also fix an arbitrary ordering of vertices of $B_x$ and define $NS_i = \{s_{j}\in B_x\ |\ j<i\}$\\

\begin{lemma}\label{lem:subroutine3}
    Let $(D,\phi)$ be an instance such that $\psi(v) < 3.75$, for every vertex in $\mathcal{U}$, let $x$ be a vertex maximizing $\mathcal{U}\cap N(x)$ and $y\in\mathcal{S}_x$. We claim $$\KFProb(D,\phi) = min\{\KFProb(D_x,\phi_x)+|N^+(x)|,\ min_{s_i}\{\KFProb(D_i,\phi_i)+|N^+(s_i)|\}\}$$ where $(D_x,\phi_x) = update(D,\phi,x,-)$ and $(D_i,\phi_i) = update(D,\phi,s_i,NS_i)$ for all $s_i\in B_x$.
\end{lemma}
\begin{proof}
    We claim that $y\in\mathcal{S}_x$ survives in every minimal solution where $x$ is not a sink. We have $y\in S_{min}$ if and only if some in-neighbour of $y$ is a sink in $D-S_{min}$, but $y\in\mathcal{S}_x$ implies  $N^-(y)\cap\mathcal{U} = \{x\}$. Hence, if $x\notin Z_{min}$, then $y\notin S_{min}$ and we can assume that if $x$ is not a sink then $y$ survives and must reach some sink. This sink by definition must be in $R^+(y)$. But if some distinct $u,v$ in $R^+(y)$ satisfy $N^+(v)\subseteq N^+(u)$, then in every minimal solution where $u$ is a sink, $v$ is also a sink and $\KFProb(update(D,\phi,v,-))\leq\KFProb(update(D,\phi,u,-))$. Thus, every possible minimal solution contains at least one element from the set $\{x\}\cup (R^+(y)\setminus D_x) = \{x\}\cup B_x$ in its sink set. Further, we can fix an arbitrary ordering for elements of $B_x$ and branch on the first vertex of this ordering which becomes a sink. In the branch where the $i^{th}$ vertex becomes a sink, we can assume that all the vertices before it are non-sinks. Hence by induction, assuming that $\KFProb(D',\phi')$ returns the optimal solution for every instance smaller than $(D,\phi)$, proves the claim.\qed
\end{proof}

\noindent Since $\psi(x)<3.75$, $N(x)$ has at most two undecided vertices whenever $x$ is undecided. We will analyze the branching vector for this subroutine in two steps. For undecided vertices with two undecided neighbours and for undecided vertices with one undecided neighbour. If $R^+(y)$ is empty, then $x$ must be a sink and no branching is involved. Also, Lemma \ref{lem:neighbourhood} implies $R^+(y)$ can have cardinality at most 2, which leads to the following possibilities. 

\subsection*{$\boldsymbol{B_x = \{s\}}$}\label{1_branch}
    Observe that if $x$ is a sink, then since $x$ has at least three vertices in its closed neighbourhood, at least two of which are undecided, we have $\psi(x)\geq 2.25$. But, if $x$ is not a sink and $s$ is the sink $y$ reaches, then since $s$ has at least two neighbours, and $x\in N^+(s)$ or $x\in R^-(s)$, we have $\psi(s)\geq\phi(N[s]\cup  \{x\})\geq 2.25$. Refer Figure \ref{S3_C1} for an illustration of this case. Hence the worst case branching vector is $(2.25,2.25)$.\\
    \begin{figure}
        \centering
        \includegraphics[height=1.25in,width=1.25in]{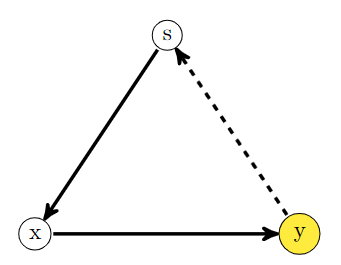}
        \caption{Case 1: $B_x = \{s\}$}
        \label{S3_C1}
    \end{figure}
    
\subsection*{$\boldsymbol{B_x = \{s_1,s_2\}}$}\label{2_branch}
    In this case, depending on the number of undecided neighbours $x$ has, we have the following cases.\\[-15pt]
    \paragraph*{\textbf{Case 1:} Branching on a vertex with two undecided neighbours}\quad\\[3pt]
    Here $|N(x)\cap\mathcal{U}| = 2$ implies $\phi(N[x]) \geq 3$ and hence $B_x\setminus N(x) = \emptyset$. Otherwise, the vertex in $B_x\setminus N(x)$ would contribute $0.75$ to $\psi(x)$ making it at least $3.75$ which contradicts our assumption that Subroutine 1 is no longer applicable. Now, assuming $y\in\mathcal{S}_x$ is the chosen vertex, we have the following possibilities.
    
    \begin{itemize}
        \item \textbf{Subcase 1.1:} $\boldsymbol{s_2\in R^-(s_1)}$ or $\boldsymbol{s_1\in R^-(s_2)}$\\[3pt]
        The arguments for both cases are identical, hence assume that  $\boldsymbol{s_2\in R^-(s_1)}$.
        \begin{itemize}
            \item[$\bullet$] If $x$ is a sink, then since $x$ has at least three undecided vertices in its closed neighbourhood each contributing 1 giving $\psi(x)\geq 3$.
            \item[$\bullet$] If $x$ is not a sink and $s_1$ is a sink which  reaches in the final solution. Then we have $x\in N(s_1)$, $s_2\in R^-(s_1)$ and hence $\psi(s_1)\geq 3$
            \item[$\bullet$] If $x$ is not a sink, $s_1$ is not a sink in the final solution but $s_2$ is. Then we have $x\in N(s_2)$, $y\in R^-(s_2)$ and $s_1\in R^+(s_2)$. Now if $s_1\in R^-(s_2)$ then $x,s_1$ and $s_2$ each contribute 1 to $\psi(s_2)$. Else, we can assume $y\neq s_1$ and $x,y,s_2$ contributes a total of 2.25 and $s_1$ contributes 0.75 giving $\psi(s_2)\geq 3$. 
        \end{itemize}
        \begin{figure}
            \centering 
            \includegraphics[height=1.25in,width=1.25in]{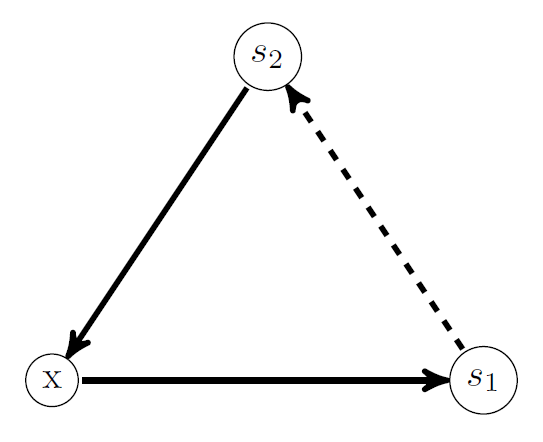}
            \caption{Case 2.1.1: $s_2\in R^-(s_1)$}
            \label{S3_C2_1_1}
        \end{figure}
        \noindent Hence we have a worst case branching vector $(3,3,3)$. Further, in the rest of the sub-cases, we can assume $s_1\notin R^-(s_2)$ and vice-versa which implies $y\notin \{s_1, s_2\}$. Now consider the case where $s_1\in N^+(x),\ s_2\in N^-(x)$, since $x\notin N^+(s_1)$, either $s_2\in N^+(s_1)$ or $s_2\in R^-(s_1)$. But $s_2\notin R^-(s_1)$ is and $s_2\in N^+(s_1)$ implies $s_1\in R^-(s_2)$, which cannot be the case. Thus this possibility is already considered, which leaves us with the following configurations.\\
        
        \item \textbf{Subcase 1.2:} $\boldsymbol{s_1, s_2\in N^-(x)}$\\[-5pt]
        \begin{itemize}
            \item[$\bullet$] If $x$ is a sink, since $N[x]$ has at least four vertices out of which three are undecided, we have $\psi(x)\geq 3.25$.
            \item[$\bullet$] If $x$ is not a sink and $s_1$ is a sink which $y$ reaches in the final solution. Then we have $x\in N^+(s_1)$, $y\in R^-(s_1)$. Also, since $N^+(s_1)\not\subseteq N^+(s_2)$ and $x\in N^+(s_2)$, $N^+(s_1)$ has at least one more vertex other than $x$, giving $\psi(s_1)\geq 2.5$.
            \item[$\bullet$] If $x$ is not a sink, $s_1$ is not a sink in the final solution, but $s_2$ is. Then we have $x\in N^+(s_2)$, $y\in R^-(s_2)$. Also, since $N^+(s_2)\not\subseteq N^+(s_1)$ and $x\in N^+(s_1)$, $N^+(s_2)$ has at least one more vertex, giving $\phi(R(s_2))\geq 2.5$. Further, since we also have the added information that $s_1$ is not a sink which gives a drop of 0.75, the total drop in potential is at least 3.25. Refer Figure \ref{S3_C2_1_2} for an illustration of this case.
        \end{itemize}
        \begin{figure}
            \centering
            \includegraphics[height=1.25in,width=1.5in]{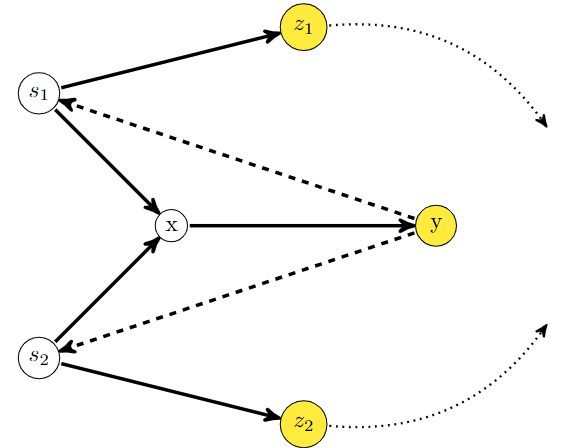}
            \caption{Case 2.1.2: $s_1, s_2\in N^-(x)$}
            \label{S3_C2_1_2}
        \end{figure}
        \noindent Here we have a worst case branching vector $(2.5,3.25,3.25)$\\
        
        \item \textbf{Subcase 1.3:} $\boldsymbol{s_1, s_2\in N^+(x)}$\\[-5pt]
        \begin{itemize}
            \item[$\bullet$] If $x$ is a sink, since $N[x]$ has at least four vertices out of which three are undecided, we have $\psi(x)\geq 3.25$.
            \item[$\bullet$] If $x$ is not a sink and $s_1$ is a sink which $y$ reaches in the final solution. Then we have $x\in N^-(s_1)$, $y\in R^-(s_1)$ and at least one other vertex in $N^+(s_1)$ giving $\psi(s_1)\geq 2.5$.
            \item[$\bullet$] If $x$ is not a sink, $s_1$ is not a sink in the final solution, but $s_2$ is. Then we have, $x\in N^-(s_2)$, $y\in R^-(s_2)$ and at least one other vertex in $N^+(s_2)$ giving $\psi(s_2)\geq 2.5$. Further, we also have the added information that $s_1$ is not a sink which gives a drop of 0.75. Hence, the total drop in potential is at least 3.25. 
        \end{itemize}
        \begin{figure}
            \centering
            \includegraphics[height=1.2in,width=2in]{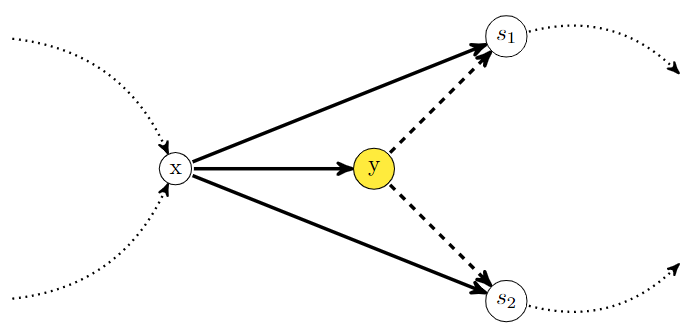}
            \caption{Case 2.1.3: $s_1,s_2\in N^+(x)$}
            \label{S3_C2_1_3}
        \end{figure}
        \noindent Here we have a worst case branching vector $(3.25,2.5,3.25)$.\\
    \end{itemize}
    \paragraph*{\textbf{Case 2:} Branching on a vertex with one undecided neighbour}\quad\\[3pt]
    Here, $|N(x)\cap\mathcal{U}| = 1$ implies $\phi(N[x]) \geq 2.25$ and hence, $|B_x\setminus N(x)| \leq 1$. Otherwise, the vertices in $B_x\setminus N(x)$ would contribute $0.75$ to $\psi(x)$ making it at least $3.75$ which contradicts our assumption that Subroutine 1 is no longer applicable. Now, assuming $s_1\in N(x)$ and $y\in\mathcal{S}_x$ is the chosen vertex, we have the following possibilities.
    
    \begin{itemize}
        \item \textbf{Subcase 2.1:} $\boldsymbol{s_1\in R^-(s_2)}$ or $\boldsymbol{s_2\in R^-(s_1)}$\\[3pt]
        The arguments for both the case are identical, hence we assume $s_1\in R^-(s_2)$. \\[-8pt]
        \begin{itemize}
            \item[$\bullet$] If $x$ is a sink, then since $x$ has at least three vertices in its closed neighbourhood exactly two of which are undecided and $s_2\in R^+(x)$, we get $\psi(x)\geq 3$.
            \item[$\bullet$] If $x$ is not a sink and $s_1$ is a sink which $y$ reaches in the final solution. Then we have at least three vertices in $N[s_1]$ out of which only $x,\ s_1$ are undecided. This and $s_2\in R^+(s_1)$ give $\psi(s_1)\geq 3$.
            \item[$\bullet$] If $x$ is not a sink, $s_1$ is not a sink in the final solution but $s_2$ is. Then we have $x\in R^-(s_2)$ and $s_1\in R^-(s_2)$. Now, $s_2$ has an out-neighbour which cannot be $x$ since $x$ has only one undecided neighbour and $s_1$ since $s_2\in R^+(s_1)$. This out-neighbour contributes at least $0.25$ making the total drop at least $3.25$. Refer Figure \ref{S3_C2_2_1_a} and \ref{S3_C2_2_1_b} for examples.\\
        \end{itemize}
        Hence the worst case branching vector is $(3,3,3.25)$.\\
        \begin{figure}
            \centering
            \begin{minipage}{.5\textwidth}
                \centering
                \includegraphics[height=1in,width=2in]{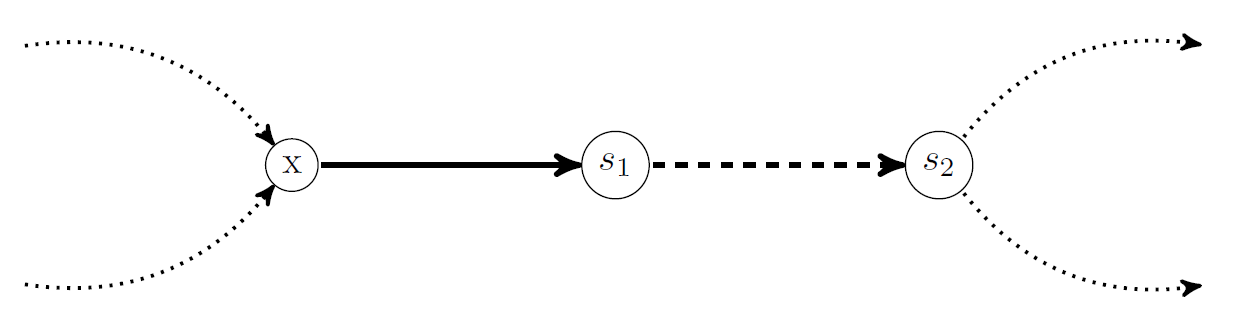}
                \caption{Case 2.2.1: $s_1\in R^-(s_2)$}
                \label{S3_C2_2_1_a}
            \end{minipage}%
            \begin{minipage}{.5\textwidth}
                \centering
                \includegraphics[height=1.25in,width=2in]{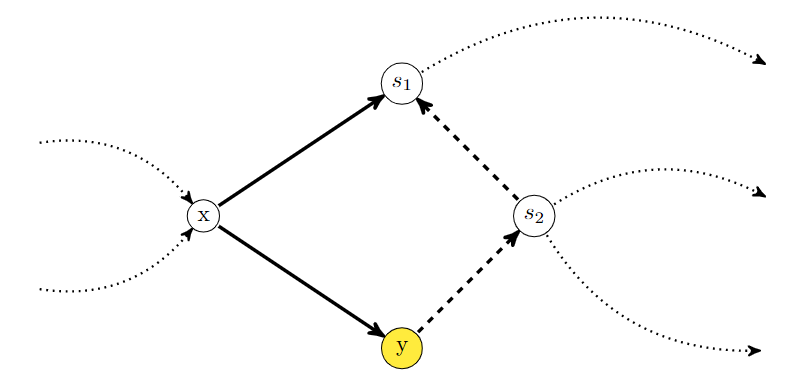}
                \caption{Case 2.2.1: $s_2\in R^-(s_1)$}
                \label{S3_C2_2_1_b}
            \end{minipage}
        \end{figure}
        
        \item \textbf{Subcase 2.2:} $\boldsymbol{s_2\notin R^-(s_1)\ \&\ s_1\notin R^-(s_2)}$\\[5pt]
        Here, $s_1\notin N^-(x)$, since in that case $s_1\in R^-(s_2)$ or $s_1\in N^+(s_2)$ which implies $s_2\in R^-(s_1)$. Also, $y\neq s_1$ since otherwise $s_1\in R^-(s_2)$.\\[-8pt]
        \begin{itemize}
            \item[$\bullet$] If $x$ is a sink, then since $x$ has at least two out-neighbours $y,s_1$ and some in-neighbour $z$, along with $s_2\in R^+(x)$ we have $\psi(x)\geq 3.25$. 
            \item[$\bullet$] If $x$ is not a sink and $s_1$ is a sink which $y$ reaches. Then we have $x,y$ in $R^-(s_1)$ and since $s_1$ has some other out-neighbour, we get $\psi(s_1)\geq 2.5$.
            \item[$\bullet$] If $x$ is not a sink, $s_1$ is not a sink in the final solution but $s_2$ is. Then we have, $x,y,z$ in $R^-(s_2)$ which gives $\psi(s_2)\geq 2.5$. Further, since we also have the added information that $s_1$ is not a sink which gives a drop of 0.75, the total drop in potential is at least 3.25. 
        \end{itemize}
        \begin{figure}
            \centering
            \includegraphics[height=1.25in,width=2in]{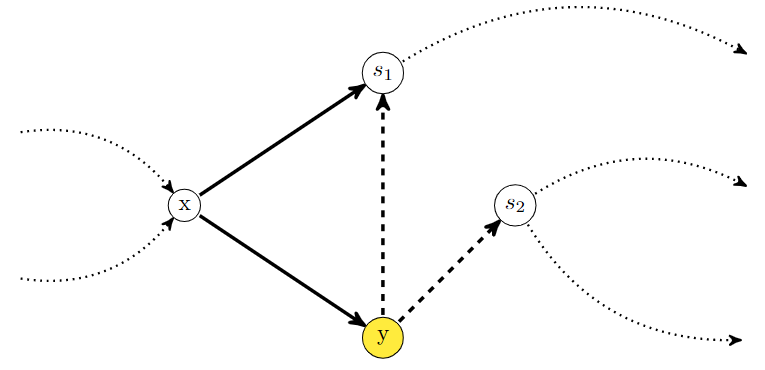}
            \caption{Case 2.2.2: $s_1,s_2\in N^+(x)$}
            \label{S3_C2_2_2}
        \end{figure}
        \noindent Hence the worst case branching vector is $(3.25,2.5,3.25)$.
    \end{itemize}
    \subsection{Correctness of Subroutine 4} \label{sec4:subsec4}
    In Subroutine $4$, we branch on undecided vertices with a non-empty surviving set and zero undecided neighbour. It is executed only when Subroutines $1,2$ and $3$ is no longer applicable. Hence, we can assume all the vertices no longer satisfy the requirements of Subroutine $1,2$ and $3$. In particular, note that every undecided vertex has zero undecided neighbours.

\begin{lemma}\label{lem:subroutine4}
    Let $(D,\phi)$ be an instance such that $\psi(v) < 3.75$ and $\mathcal{U}\cap N(v) = \emptyset$ for every vertex in $\mathcal{U}$. Let $x$ be an undecided vertex, $y\in\mathcal{S}_x$ and $s\in R^+(y)$. We claim $$\KFProb(D,\phi) = min\{\KFProb(D_x,\phi_x)+|N^+(x)|,\ \KFProb(D_s,\phi_s)+|N^+(s)|\}$$ where $(D_x,\phi_x) = update(D,\phi,x,-)$ and $(D_s,\phi_s) = update(D,\phi,s,-)$.
\end{lemma}
\begin{proof}
    We claim that $y\in\mathcal{S}_x$ survives in every minimal solution where $x$ is not a sink. We have, $y\in S_{min}$ if and only if some in-neighbour of $y$ is a sink in $D-S_{min}$, but $y\in\mathcal{S}_x$ implies  $N^-(y)\cap\mathcal{U} = \{x\}$. Hence, if $x\notin Z_{min}$, then $y\notin S_{min}$. Now, assuming $x$ is not a sink, $y$ survives, and it must reach some sink, which by definition must belong to $R^+(y)$. Now assume that $s\in R^+(y)$ is not a sink in the final solution where $x$ is not a sink and $y$ survives. Since $s$ is not a sink, it must either belong to the solution set or reach a sink vertex. But $s$ cannot belong to any minimal feasible solution to $(D,\phi)$, since $N^-(s)\cap\mathcal{U} = \emptyset$. Hence we can assume $s$ reaches some sink $z$. Now there are two possibilities  either $N^+(z)$ intersects a vertex of $P_{(y,s)}$ or not.
    See Figure \ref{S4_C1} and Figure \ref{S4_C2} in appendix for an illustration of the following cases:    
    \begin{itemize}
        \item \textbf{Case 1:} $\boldsymbol{w_1\in N^+(z)\cap P_{(y,s)}}$. Here, $z$ has a path to $s$ in $D-N^+(s)$ via $P_{(w_1,s)}$. Also, $x,z\notin N(s)$ since $s$ is undecided and cannot have undecided neighbours and hence $x\in R^-(s)$ since it has a path to $s$ via $y$. Now, $w_1\neq y$, since $z$ is a sink that $s$ can reach in a solution where $y$ survives and hence $y,w_1\in R^-(s)$. Thus, $R^-(s)$ contains $x,y,w_1,z$ and $s$. Now, out-neighbour of $s$ in $P_{(s,z)}$ cannot be $y$ or $w_1$ since $N^+(s)$ does not intersect $P_{(y,s)}$ and $y,w_1\in P_{(y,s)}$. Thus $N^+(s)$ contains at least one semidecided vertex $w_2$ different from $y,w_1$ which implies that $x,y,w_1,z,w_2$ and $s$ belongs to $R(s)$. This leads to a contradiction since this implies $\phi(R(s))\geq 3.75$.\\
        \begin{figure}
            \centering
            \includegraphics[height=1.3in,width=3.25in]{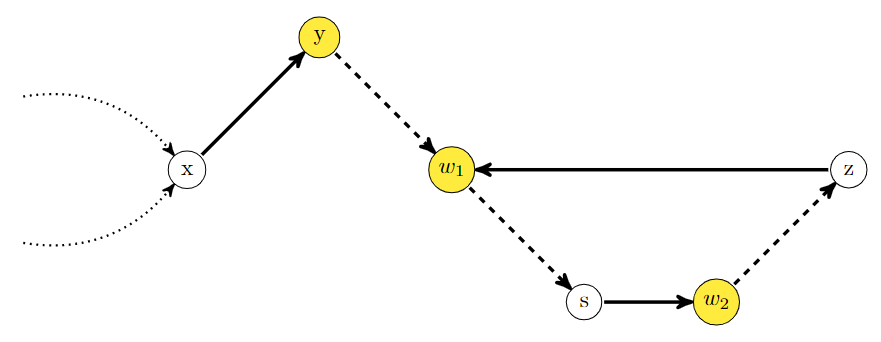}
            \caption{Case 1: $w_1\in N^+(z)\cap P_{(y,s)}$}
            \label{S4_C1}
        \end{figure}
        
        \item \textbf{Case 2:} $\boldsymbol{N^+(z)\cap P_{(y,s)}=\emptyset}$. Here, $y$ can reach $z$ in $D-N^+(z)$ via the path $P_{(y,s)}$ followed by $P_{(s,z)}$. Observe that $x\notin N^+(z)$ as both are undecided, and $y\in R^-(z)$ which implies that $x\in R^-(z)$. Now, out-neighbour of $s$ in $P_{(s,z)}$ cannot be $y$ since $N^+(s)$ does not intersect $P_{(y,s)}$ and $y\in P_{(y,s)}$. Thus $P_{(s,z)}$ contains at least one semidecided vertex $w_1$ different from $y$ and $R^-(z)$ contains $x,y,w_1,z$ and $s$. Now, out-neighbour of $z$ cannot be $y$, since $z$ is a sink that $s$ can reach in a solution where $y$ survives and it cannot be $w_1$ since $N^+(z)$ does not intersect $P_{(s,z)}$ and $w_1\in P_{(s,z)}$. Thus $N^+(z)$ contains at least one semidecided vertex $w_2$ different from $y,w_1$ which implies that $x,y,w_1,z,w_2$ and $s$ belong to $R(z)$. This leads to a contradiction since, $\phi(R(z))\geq 3.75$.\\
        \begin{figure}
            \centering
            \includegraphics[height=1.3in,width=3.25in]{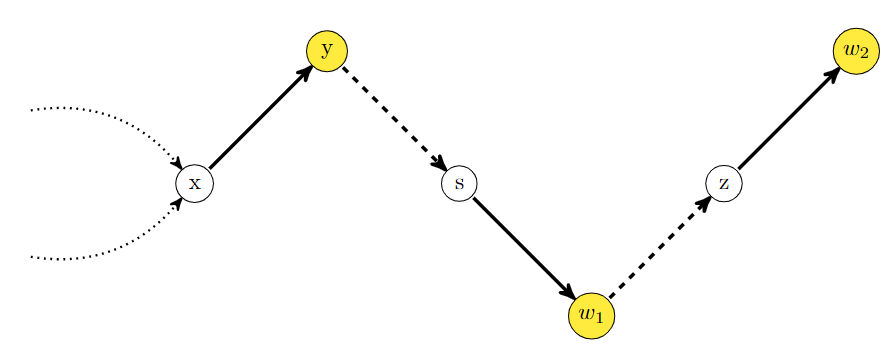}
            \caption{Case 2: $N^+(z)\cap P_{(y,s)}=\emptyset$}
            \label{S4_C2}
        \end{figure}
    \end{itemize}
    
    \noindent Observe that both possibilities lead to a contradiction and hence our assumption that $s$ reaches some sink $z$ in a minimal solution where $x$ is not a sink and $y$ survives, has to be wrong. Thus, for the given instance where Subroutines $1,2$ and $3$ are no longer applicable, in any minimal solution where $x$ is not a sink, $y\in\mathcal{S}_x$ and $s\in R^+(y)$, $s$ has to be a sink. By induction, assuming that $\KFProb(D',\phi')$ returns the optimal solution for every instance smaller than $(D,\phi)$, proves the claim.\qed
\end{proof}

\noindent If $R^+(y)=\emptyset$ then $x$ has to be a sink and the subroutine is executed as a reduction rule without branching. Now, since $x$ has at least two neighbours and as $s\in R^+(x)$ we get $\psi(x)\geq 2.25$. Similarly, since $s$ has at least two neighbours and $x\in R^-(s)$ we get $\psi(s)\geq 2.5$. Hence the worst case branching factor is $(2.25,2.5)$.
    
\section{Running time analysis}
Observe that the reduction Rules $1$, $2$ and $3$ can be applied on any input instance in polynomial time. The branching vector obtained in Subroutine $1$ was $(3.75,0.75)$, which gives the recurrence $f(\mu)\leq f(\mu-3.75)+f(\mu-0.75)$. This solves to $f(\mu)=\mathcal{O}(1.4549^\mu)$. For Subroutine $2$ we observe a worst case drop in measure of $(2.25,3)$ and $(3,3,3)$, amongst which $(3,3,3)$ has the higher running time $f(\mu)=\mathcal{O}(1.4422^\mu)$. Similarly, for Subroutine $3$, we have branching vectors $(2.25,2.25),\ (3.25,2.5,3.25)$, and $(3,3,3)$, out of which $(3.25,2.5,3.25)$ gives the worst running time $f(\mu)=\mathcal{O}(1.4465^\mu)$. Finally Subroutine $4$ has a branching vector $(2.25,2.5)$ which gives $f(\mu)=\mathcal{O}(1.3393^\mu)$.\\

\noindent For an input instance $D$ of {\sc Knot-Free Vertex Deletion}, we initialise the potential of each vertex to $1$ and run KFVD($D,\phi$). Observe that, $\phi(V(D)) = \lvert V(D)\rvert = n$ and during the recursive calls the potential of the instance never increases or drops below 0. Hence we can bound the running time of the algorithm by the run time corresponding to its worst case subroutine (here, Subroutine 1), which gives us the following theorem.\\

\begin{theorem}
 Algorithm ${\KFProb}$ solves {\sc Knot-Free Vertex Deletion}  in $\mathcal{O}^*(1.4549^n)$.
\end{theorem}

\section{A lower bound on running time}
In this section, we construct a family of graphs and subsequently prove a lower bound on the running time of our algorithm.\\ \vspace{-10pt}

\begin{figure}
    \begin{center}
        \includegraphics[height=1.5in,width=4.5in]{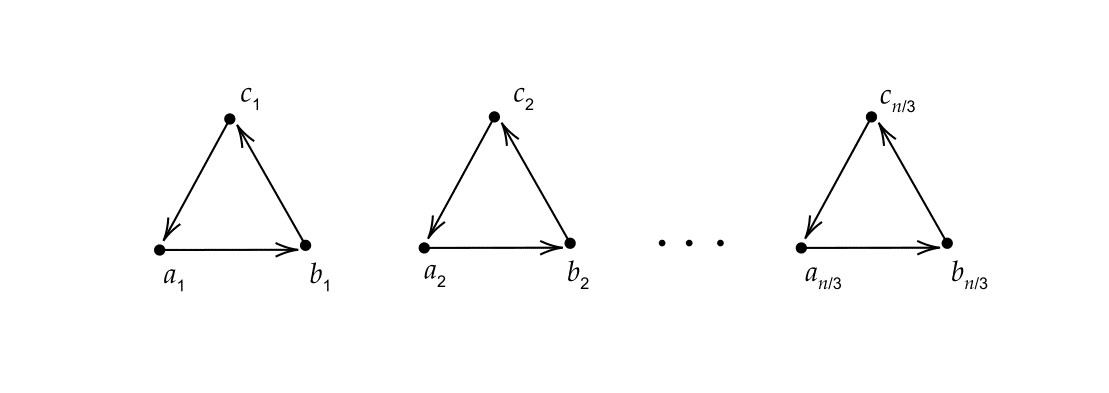}\vspace{-10pt}
        \caption{Illustration of a worst-case instance for our algorithm.}\label{fig:lowerbound}
    \end{center}
\end{figure}\vspace{-10pt}

\noindent We run our algorithm $\KFProb$ on the graph $D$ (Figure \ref{fig:lowerbound} in the appendix), where $V(D)=\{a_i,b_i,c_i\mid 1\leq i\leq \frac{n}{3}\}$ and $E(D)=\{(a_i,b_i),(b_i,c_i),(c_i,a_i)\mid 1\leq i\leq \frac{n}{3}\}$. We claim that in the worst case, our algorithm takes $\mathcal{O}^*(3^{n/3})$ time to solve \KFProb~on $D$ which we prove via adversarial arguments. Initially since $\phi(v) = 1$ for every vertex of $D$, we have $\psi(a_i)=3$, $\psi(b_i)=3$, $\psi(c_i)=3$. Since $\psi(v)<3.75$ and $\mathcal{S}_{v}\neq\emptyset$ for every vertex, Subroutine 1 and 2 are not applicable. The adversary chooses the vertex $a_i$ to branch on. Since $b_i\in\mathcal{S}_{a_i}$ and $b_i,c_i\in R^+(b_i)$ we get one branch where $a_i$ becomes a sink with potential drop 3. Another where $a_i$ is not a sink, but $b_i$ is, with drop 3. Finally, one where $a_i,b_i$ are not sinks, but $c_i$ is with drop 3. In all of the above branches, $\{a_i,b_i,c_i\}$ is removed, giving a recurrence relation $T(n)=3T(n-3)$ which implies $T(n)$ is $3^{n/3}\geq 1.4422^n$ for this instance. Which gives us the following theorem.\\

\begin{theorem}
    Algorithm $\KFProb$ runs in time $\Omega(1.4422^n)$.
\end{theorem}

\section{Upper \& lower bounds for the number of minimal knot-free vertex deletion sets}
We claim that if we run our algorithm on any given directed graph, and create a decision tree then for every minimal knot-free vertex deletion set there exists a leaf node of the decision tree which corresponds to it. Note that any algorithm which finds all the minimal solutions needs at least unit time to find each solution and hence the number of minimal solutions for any graph of size $n$ cannot exceed the complexity of the algorithm. This observation gives us the following theorem.
 
\begin{theorem}
    The number of inclusion-wise minimal knot-free vertex deletion sets is $\mathcal{O}^*(1.4549^n)$.
\end{theorem}

\begin{proof}
    Let $S_{min}$ be a minimal solution and $Z_{min}$ be the set of sinks in $D-S_{min}$. Beginning with the root of decision tree use the following set of rules to find the corresponding leaf node. If the node corresponds to an execution of subroutine 1 on a vertex $v$, then if $v\in Z_{min}$ choose the branch of the tree where $v$ is added to the sink set, else choose the branch where $v$ is labelled as a non-sink vertex.  If the node corresponds to any other subroutine, then by correctness of the algorithm proven in the earlier section, at that node it branches on a set of vertices, which intersects the sink set of any minimal solution. Choose a branch corresponding to a vertex $s$ such that $s\in Z_{min}$. Follow this procedure to get to a leaf node which corresponds to a knot-free vertex deletion set $S_{leaf}$ and sink set $Z_{leaf}$. Note that due to our choice of leaf node, we have $Z_{leaf}\subseteq Z_{min}$ and consequently $S_{leaf}\subseteq S_{min}$. This along with minimality of $S_{min}$ gives $S_{leaf}=S_{min}$. Hence the number of minimal knot-free vertex deletion sets of $D$ cannot exceed the number of leaf nodes of the decision tree corresponding to any run of the algorithm on $D$. Hence maximum number of minimal knot-free vertex deletion sets is $\mathcal{O}^*(1.4549^n)$.\qed
\end{proof}

\begin{theorem}
    There exists an infinite family of graphs with $\Omega(1.4422^n)$ many inclusion-wise minimal knot-free vertex deletion sets.
\end{theorem}

\begin{proof}
    Consider the graph in Figure \ref{fig:lowerbound}. Each strongly connected component $\{a_i,b_i,c_i\}$ can be made knot-free by deleting a single vertex. Hence every set $S$ which contains only one element from $\{a_i,b_i,c_i\}$ is a knot-free vertex deletion set. Observe that there are $3^{\frac{n}{3}}$ many of them since we can choose the element for each $i$ in 3 ways and $i$ ranges from $1$ to $\frac{n}{3}$. Further, any proper subset $S'$ of such a set will not intersect $\{a_i,b_i,c_i\}$ for some $i$, leaving $D-S'$ with at least one knot. Hence the graph in Figure \ref{fig:lowerbound} has at least $3^{\frac{n}{3}}\geq 1.4422^n$ many minimal knot-free vertex deletion sets. Now, by taking graphs which are disjoint union of triangles, we obtain an infinite family of graphs such that each element of that family has at least $1.4422^n$ many minimal knot-free vertex deletion sets.\qed
\end{proof}

\section{Conclusion}
    
We obtain a $\mathcal{O}(1.4549^n)$ time algorithm for the \KFProb problem which uses polynomial space. We also obtain an upper bound of $\mathcal{O}(1.4549^n)$ on the number of minimal knot-free vertex deletion sets possible for any directed graph and present a family of graphs which have $1.4422^n$ many minimal knot-free vertex deletion sets. Our algorithm is not proven to be optimal and improving it is a possible direction for future work. Closing the gap between the upper and lower bound for the maximum number of knot-free vertex deletion sets is also of interest.

\section*{Acknowledgements} 
\noindent We are thankful to Ajinkya Gaikwad for useful discussions and his comments on  Algorithm \ref{algo1}. Ajaykrishnan E S would like to thank DST-INSPIRE for their support via the Scholarship for Higher Education (SHE) programme.
\bibliography{my}

\begin{thebibliography}{10}

\bibitem{knot_finegrained}
Alan~Diêgo Aurélio~Carneiro, Fábio Protti, and Uéverton Souza.
\newblock On knot-free vertex deletion: Fine-grained parameterized complexity
  analysis of a deadlock resolution graph problem.
\newblock {\em Theoretical Computer Science}, 909, 01 2022.

\bibitem{DBLP:conf/iwpec/BessyBCPS19}
St{\'{e}}phane Bessy, Marin Bougeret, Alan Di{\^{e}}go~A. Carneiro, F{\'{a}}bio
  Protti, and U{\'{e}}verton~S. Souza.
\newblock Width parameterizations for knot-free vertex deletion on digraphs.
\newblock In Bart M.~P. Jansen and Jan~Arne Telle, editors, {\em 14th
  International Symposium on Parameterized and Exact Computation, {IPEC} 2019,
  September 11-13, 2019, Munich, Germany}, volume 148 of {\em LIPIcs}, pages
  2:1--2:16. Schloss Dagstuhl - Leibniz-Zentrum f{\"{u}}r Informatik, 2019.

\bibitem{bjrklund:LIPIcs:2017:6948}
Andreas Bj{\"o}rklund.
\newblock {Determinant Sums for Hamiltonicity (Invited Talk)}.
\newblock In Jiong Guo and Danny Hermelin, editors, {\em 11th International
  Symposium on Parameterized and Exact Computation (IPEC 2016)}, volume~63 of
  {\em Leibniz International Proceedings in Informatics (LIPIcs)}, pages
  1:1--1:1, Dagstuhl, Germany, 2017. Schloss Dagstuhl--Leibniz-Zentrum fuer
  Informatik.

\bibitem{10.1007/s10878-018-0279-5}
Alan~Di\^{e}go Carneiro, F\'{a}bio Protti, and U\'{e}verton~S. Souza.
\newblock Deadlock resolution in wait-for graphs by vertex/arc deletion.
\newblock {\em J. Comb. Optim.}, 37(2):546–562, feb 2019.

\bibitem{10.1145/1374376.1374404}
Jianer Chen, Yang Liu, Songjian Lu, Barry O'Sullivan, and Igor Razgon.
\newblock A fixed-parameter algorithm for the directed feedback vertex set
  problem.
\newblock In {\em Proceedings of the Fortieth Annual ACM Symposium on Theory of
  Computing}, STOC '08, page 177–186, New York, NY, USA, 2008. Association
  for Computing Machinery.

\bibitem{diestel-book}
Reinhard Diestel.
\newblock {\em Graph Theory}.
\newblock Springer Publishing Company, Incorporated, 5th edition, 2017.

\bibitem{10.1145/2897518.2897551}
Fedor~V. Fomin, Serge Gaspers, Daniel Lokshtanov, and Saket Saurabh.
\newblock Exact algorithms via monotone local search.
\newblock In {\em Proceedings of the Forty-Eighth Annual ACM Symposium on
  Theory of Computing}, STOC '16, page 764–775, New York, NY, USA, 2016.
  Association for Computing Machinery.

\bibitem{10.1007/11602613_58}
Fedor~V. Fomin, Fabrizio Grandoni, Artem~V. Pyatkin, and Alexey~A. Stepanov.
\newblock Bounding the number of minimal dominating sets: A measure and conquer
  approach.
\newblock In {\em Proceedings of the 16th International Conference on
  Algorithms and Computation}, ISAAC'05, page 573–582, Berlin, Heidelberg,
  2005. Springer-Verlag.

\bibitem{fomin_exact}
Fedor~V. Fomin and Dieter Kratsch.
\newblock {\em Exact Exponential Algorithms}.
\newblock Springer-Verlag, Berlin, Heidelberg, 1st edition, 2010.

\bibitem{10.1145/800029.808532}
Michael Held and Richard~M. Karp.
\newblock A dynamic programming approach to sequencing problems.
\newblock In {\em Proceedings of the 1961 16th ACM National Meeting}, ACM '61,
  page 71.201–71.204, New York, NY, USA, 1961. Association for Computing
  Machinery.

\bibitem{10.1007/978-3-030-64843-5_21}
Gordon Hoi.
\newblock An improved exact algorithm for the exact satisfiability problem.
\newblock In Weili Wu and Zhongnan Zhang, editors, {\em Combinatorial
  Optimization and Applications}, pages 304--319, Cham, 2020. Springer
  International Publishing.

\bibitem{lima2018and}
Carlos~VGC Lima, F{\'a}bio Protti, Dieter Rautenbach, U{\'e}verton~S Souza, and
  Jayme~L Szwarcfiter.
\newblock And/or-convexity: a graph convexity based on processes and deadlock
  models.
\newblock {\em Annals of Operations Research}, 264:267--286, 2018.

\bibitem{10.1145/3446969}
Daniel Lokshtanov, Pranabendu Misra, Joydeep Mukherjee, Fahad Panolan,
  Geevarghese Philip, and Saket Saurabh.
\newblock 2-approximating feedback vertex set in tournaments.
\newblock {\em ACM Trans. Algorithms}, 17(2), apr 2021.

\bibitem{doi:10.1137/1.9781611976465.14}
Daniel Lokshtanov, Pranabendu Misra, M.~S. Ramanujan, Saket Saurabh, and Meirav
  Zehavi.
\newblock {\em FPT-approximation for FPT Problems}, pages 199--218.
\newblock Association for Computing Machinery, 2021.

\bibitem{moon_moser}
J.~W. Moon and L.~Moser.
\newblock On cliques in graphs.
\newblock {\em Israel Journal of Mathematics}, 3(1):23--28, 1965.

\bibitem{https://doi.org/10.1002/net.21537}
Fabiano de~S. Oliveira and Valmir~C. Barbosa.
\newblock Revisiting deadlock prevention: A probabilistic approach.
\newblock {\em Networks}, 63(2):203--210, 2014.

\bibitem{ramanujan2022exact}
MS~Ramanujan, Abhishek Sahu, Saket Saurabh, and Shaily Verma.
\newblock An exact algorithm for knot-free vertex deletion.
\newblock In {\em 47th International Symposium on Mathematical Foundations of
  Computer Science (MFCS 2022)}, 2022.

\bibitem{Razgon2007ComputingMD}
Igor Razgon.
\newblock Computing minimum directed feedback vertex set in
  $\mathcal{O}(1.9977^n)$.
\newblock In {\em Italian Conference on Theoretical Computer Science}, 2007.

\end{thebibliography}

\end{document}